\documentclass[twoside]{article}

\usepackage[accepted]{aistats2026}
\makeatletter
\def\@copyrightspace{%
  \@float{copyrightbox}[b]%
  \noindent\rule{0.8in}{0.4pt}\par
  \vskip 3pt
  \ifnum\statePaper=1
    {\small\noindent$^{*}$Correspondence to: \texttt{thandina@caltech.edu}.\par}%
    \vskip 2pt
  \fi
  {\small\noindent\Notice@String\par}%
  \vskip 2pt
  \end@float}
\makeatother

\usepackage[utf8]{inputenc} 
\usepackage[T1]{fontenc}    
\usepackage{hyperref}       
\usepackage{url}            
\usepackage{booktabs}       
\usepackage{amsfonts}       
\usepackage{nicefrac}       
\usepackage{microtype}      
\usepackage{xcolor}         
\usepackage{microtype}
\usepackage{graphicx}
\usepackage{subcaption}
\usepackage{enumitem}
\usepackage{algpseudocode}
\usepackage{booktabs}
\usepackage[ruled,vlined]{algorithm2e}
\usepackage{amsmath}

\usepackage{amsmath,amsfonts,bm}
\usepackage{enumitem}

\def\eqref#1{equation~\ref{#1}}

\def\1{\bm{1}}

\DeclareMathAlphabet{\mathsfit}{\encodingdefault}{\sfdefault}{m}{sl}
\SetMathAlphabet{\mathsfit}{bold}{\encodingdefault}{\sfdefault}{bx}{n}

\newcommand{\E}{\mathbb{E}}

\newcommand{\poly}{\mathtt{poly}}

\DeclareMathOperator*{\argmax}{arg\,max}

\usepackage{url}
\usepackage[round]{natbib}
\usepackage{hyperref}

\usepackage{amsmath}
\usepackage{amssymb}
\usepackage{mathtools}
\usepackage{amsthm}

\usepackage[capitalize,noabbrev]{cleveref}

 \theoremstyle{plain}
\newtheorem{theorem}{Theorem}[section]
\newtheorem{proposition}[theorem]{Proposition}

\newtheorem{corollary}[theorem]{Corollary}

\theoremstyle{definition}
\newtheorem{definition}[theorem]{Definition}
\newtheorem{assumption}[theorem]{Assumption}
\theoremstyle{remark}

\usepackage[textsize=tiny]{todonotes}

\author{
  Tinashe Handina \\
  Department of Computing + Mathematical Sciences\\
  California Institute of Technology\\
  Pasadena, CA 91125 \\
  \texttt{thandina@caltech.edu} \\
  \And
  Tongxin Li \\
  School of Data Science \\
  The Chinese University of Hong Kong \\
  Shenzhen, China\\ 
  \texttt{litongxin@cuhk.edu.cn} \\
  \And
  Kishan Panaganti \\
  Department of Computing + Mathematical Sciences\\
  California Institute of Technology\\
  Pasadena, CA 91125 \\
  \texttt{kbp@caltech.edu} \\ 
  \And
  Eric Mazumdar \\
  Department of Computing + Mathematical Sciences\\
  California Institute of Technology\\
  Pasadena, CA 91125 \\
  \texttt{mazumdar@caltech.edu} \\
  \And
  Adam Wierman \\
  Department of Computing + Mathematical Sciences\\
  California Institute of Technology\\
  Pasadena, CA 91125 \\
  \texttt{adamw@caltech.edu} \\
}

\begin{document}
\twocolumn[

\aistatstitle{Leveraging Machine-Learned Advice in Strategic Interactions with No-Regret Learners}

\aistatsauthor{ Tinashe Handina$^{1, *}$\And Tongxin Li$^2$ \And  Kishan Panaganti$^1$ \AND Eric Mazumdar$^1$ \And Adam Wierman$^1$ }

\aistatsaddress{ $^1$California Institute of Technology
\And  $^2$The Chinese University of Hong Kong } ]
\begin{abstract}
  We study how an agent in a two-player repeated game can effectively utilize potentially imperfect advice when interacting with a no-regret learner. We characterize the advice landscape by introducing a pseudo-metric to quantify the usefulness of an advice instance.  We demonstrate the pseudo-metric's applicability through two forms of advice: simulators and payoff matrix predictions. We then show how an optimizing player, equipped with correctness guarantees on the advice, could leverage simulators to compute approximate Stackelberg strategies more efficiently, reducing the interaction complexity traditionally required and illustrating the power of good advice. Finally, we extend our analysis to settings where the advice does not have any guarantee of correctness. We find that, in general, a player cannot simultaneously guarantee near Stackelberg performance when the advice is approximately accurate and a no-regret condition when the advice is inaccurate. We do show, however, that it is possible for an advice-aided player to weakly dominate their utility in some (coarse)-correlated equilibria.
\end{abstract}

\section{INTRODUCTION}
\label{sec: introduction}

Learning in games, a fundamental area within game theory and artificial intelligence, explores how players adapt their strategies through interactions to achieve their objectives~\citep{fudenberg2009learning, leyton2008essentials}.  This field has seen significant developments, with models like fictitious play and no-regret learning providing frameworks for understanding strategic adaptation. No-regret algorithms, such as multiplicative weights update or follow the perturbed leader, guarantee that a player's average payoff over $T$ time steps will be asymptotically no worse than that of the best single fixed action chosen in hindsight (minimizing external regret)~\citep{arora2005multiplicative, cesa2006prediction}.  
Stronger guarantees, like minimizing internal or swap regret, compare performance against policies that swap actions. Algorithms minimizing swap regret are particularly significant as their empirical play frequencies converge to the set of correlated equilibria when employed by all players~\citep{blum2007external, hart2000simple}. The advent of online learning and deep reinforcement learning has further amplified these concepts, enabling sophisticated applications such as human-computer game AIs where agents dynamically adapt to complex, non-stationary environments~\citep{ma2024large, zhang2024guide}.




Traditional analysis of learning in games, particularly in repeated games, often assumes players rely solely on their own observations and historical data, as seen in classical adaptive strategies~\citep{freund1999adaptive}. However, the rise of AI has shifted this paradigm, placing external advice such as predictions or recommendations from machine learning models, expert systems, or other agents, at the forefront of strategic decision-making~\citep{silver2016mastering}. This oversight is significant, as the integration of Machine Learning (ML)-driven advice naturally introduces a departure from the self-contained learning paradigms traditionally studied, necessitating a re-evaluation of how strategies are formed and optimized when agents can leverage advice to infer opponents’ behaviors or payoff structures.

In modern strategic interactions, particularly those involving AI agents, the role of external `advice' as insights into an opponent’s strategy or the game environment, is increasingly critical. For example, in strategic games, opponent modeling techniques~\citep{yu2022model,yu2025llm} enable advanced AI systems to predict opponent moves and analyze game states. Similarly, in interactive puzzle solving, AI-driven systems generate adaptive hints that guide users through challenging tasks more efficiently~\citep{mozikov2024eai}. Furthermore, for market dynamics, ML methods  provide crucial advisory functions to inform optimal bidding or pricing strategies~\citep{deng2024autobidder}, extending traditional analyses of high-stakes scenarios like spectrum auctions~\citep{cramton2013spectrum}. Advice can also aid in the discovery of optimal strategies. For instance, in repeated games, Stackelberg strategies have been found to be optimal up to $o(T)$ additive factors against no-swap-regret learners \citep{Deng2019StrategizingAN} and have thus been a desirable strategy even when facing no-regret learners\citep{zhang2025learningsteerlearnersgames}. 
However, against certain no-regret algorithms (e.g., mean-based algorithms), an optimizer requires an exponential number of interactions (with respect to the size of its action space), to determine an $\varepsilon$-approximate Stackelberg strategy~\citep{brown2023learninggamesgoodlearners}. This complexity arises because the optimizer must extensively explore the learner’s behavior to devise an effective strategy. Intuitively, advice about the learner’s tendencies or decision-making process could reduce this exploration burden, yet whether it can improve interaction complexity remains unresolved. Moreover, advice can be pivotal in scenarios where the payoff matrix is not fully known. In games with incomplete information, predictions about the payoff matrix can significantly influence strategy selection, as demonstrated in~\citep{fudenberg2019predicting}, where accurate payoff predictions led to more effective strategies. 

Recent pioneering analysis has investigated Pareto optimality in normal-form games and stochastic Bayesian games given strategy beliefs~\citep{lisafe}; however, the incorporation of ML-generated advice into repeated games remains largely underexplored, primarily due to challenges in ensuring its quality and robustness. Inaccuracies may stem from computational limitations in modeling dynamic environments~\citep{de2018multi}, incomplete or noisy training data, lack of model interpretability that obscures the basis of advice~\citep{xie2024can}, emotional LLM prompts~\citep{mozikov2024eai},  and, in adversarial settings, the risk of deliberately misleading inputs. These issues highlight the need for understanding the role of advice in repeated games. Thus, this paper addresses the question: 

\begin{center}
\textit{How can an optimizer play a repeated game against a no-regret learner when provided with ML advice that may contain errors?}
\end{center}

We answer this by systematically exploring the role of advice in learning within repeated games, focusing on technical advancements that enhance strategic decision-making under uncertainty. Building on the legacy of algorithmic game theory, we adapt it to an AI-driven era where external inputs redefine optimal strategies. 


\textbf{Contributions}:
First, to characterize the quality of advice, we introduce a \textit{pseudo-metric} (\Cref{def:advice_quality}) measuring the discrepancy between advice and the learner's true behavior or payoffs. 
We demonstrate that access to high-quality advice enables efficient computation of approximate Stackelberg strategies. In particular, in~\Cref{thm:efficient_computation} we show how in a large set of games with sufficiently accurate advice and against classes of no-regret learners, the optimizer requires only a polynomial number of interactions $\poly(|\mathcal{A}|,|\mathcal{B}|,1/\varepsilon)$ where $\mathcal{A}$ and $\mathcal{B}$ are discrete action sets of the optimizer and learner respectively to identify the $\varepsilon$-approximate Stackelberg strategy, compared to the exponential number needed without such guidance \citep{brown2023learninggamesgoodlearners}.


As the accuracy of the advice is, in general, not guaranteed, the optimizer faces a challenge: leveraging potentially flawed advice while maintaining performance. Focusing on payoff matrix predictions, we explore whether the optimizer can achieve both near-Stackelberg performance when advice is accurate (above some correctness threshold) and simultaneously maintain a standard no-regret property when advice is inaccurate (below some correctness threshold).
We find that this `best of both worlds' dual guarantee is generally unattainable (see~\Cref{thm:impossibility}). Despite this impossibility result, we establish that an optimizer can leverage an advice-augmented algorithm that is able to strictly dominate their utility under some (coarse)-correlated equilibria for some learner instances, whilst guaranteeing weak domination of their utility under some (coarse)-correlated equilibria in general (see~\Cref{thm:no-regret}). 


\subsection{Related Work}
\textbf{Learning Types and Information in Games:} Learning and trying to utilize information about an unknown player dates back to concepts of type based methodologies linked to Bayesian games. Leveraging types to model uncertainty over players can be traced to works in the late 1960s \citep{harsanyi1967games}, \citep{harsanyi1968games}. Subsequent pioneering work \citep{milgrom1991adaptive} included ideas on how to learn and adapt beliefs in a bid to better understand how a player can make use of these beliefs. However, a range of works have since demonstrated the challenges that players face when attempting to make correct predictions whilst aiming for optimal play \citep{nachbar2005beliefs, foster2001impossibility, dekel2004learning}. The theoretical foundations in these works have since been applied to many real world domains. In particular, \citep{southey2012bayes} leveraged type based methods in poker where players' hands are partially hidden. Recent work \citep{balcan2025structured, EshwarLearningtoPlay25} extends learning types to structured games or learning to play against opponents when drawn from distributions. While these works assume some structure with respect to the opponent, our learner receives arbitrary, potentially untrustworthy advice about the opponent, with no a priori guarantees. This advice need not correspond to any distribution over opponents, any model of payoff uncertainty, nor any coherent generative assumption. 

\textbf{Online Decision Making with Predictions:} The study of advice-augmented algorithms has, over the years, grown in prominence given the rapid proliferation and adoption of machine learning algorithms across a wide range of domains. \citep{MahdianOLP} sought to understand how to best design advice-augmented algorithms for problems such as online ad allocation and load balancing. Spurring the increased development and deployment of advice-augmented algorithms is their particular salience in a wide range of practical domains such as wind farm optimization \citep{elkin_witherspoon_wind_energy}. These algorithms' reach has extended to settings like competitive caching \citep{pmlr-v80-lykouris18a}, energy scheduling \citep{lee_energy}, network flow allocation \citep{lavastida_et_al:LIPIcs.ESA.2021.59}, smoothed online optimization \citep{RuttenSOPO}, and games \citep{lisafe}. A key design characteristic of this setting is that the advice typically has no guarantees on correctness and is assumed to come from some black-box. The main goal in this setting is to find algorithms that effectively leverage advice. This effectiveness is often characterized by achieving an optimal trade-off between attaining near-optimal performance when the advice is correct (\textit{consistency}) and having some guarantees on performance when the advice is incorrect (\textit{robustness}). 

\textbf{No-regret Learning in Games:}
There is a long history of work in game theory seeking to understand the characteristics of no-regret play. Prominently, there has been work devoted to understanding the convergence of players with different no-regret guarantees in a wide range of game contexts \citep{pmlr-v32-krichene14, leme2024convergence}. Building on this, there has also been work investigating trajectory dynamics of no-regret play with comparison to rewards obtained under Nash equilibria\citep{Katrina_beating_Nash_NR}. Particularly related to this paper, there has been a line of work trying to understand how a player may optimize their reward when in a game with an opponent satisfying some no-regret guarantees \citep{Selling_NR_buyer, pmlr-v178-mansour22a, Deng2019StrategizingAN, brown2023learninggamesgoodlearners}. Prior work has shown how it is impossible to steer no-regret learners to Stackelberg equilibria without additional information on their algorithm or objectives \citep{zhang2025learningsteerlearnersgames}. Related work has also explored how providing advice or recommendations to agents can influence equilibrium outcomes \citep{balcan2009improved, balcan2013circumventing}. While these papers consider the problem of designing helpful advice and convincing agents to follow it, we build upon this work by examining how an agent provided advice with no guarantees can incorporate it to optimize against no-regret learners.

\vspace{-5pt}

\section{PRELIMINARIES}
\label{sec: preliminaries}

\vspace{-5pt}
We consider a two-player repeated game defined by a bimatrix $G = (\mathbf{A}, \mathbf{B}) $, where Player A (the optimizer) has a finite action set $\mathcal{A} = \{ a_1, \ldots, a_N \}$ and Player B (the learner) has a finite action set $\mathcal{B} = \{ b_1, \ldots, b_M \}$. Each player knows their own utility matrix ($\mathbf{A}$ for Player A, $\mathbf{B}$ for Player B) but not their opponent's. The game proceeds over a finite time horizon $t \in [T] $, with players simultaneously selecting mixed strategies at each step: $\alpha_t \in \Delta(\mathcal{A})$ for Player A and $\beta_t \in \Delta(\mathcal{B})$ for Player B. The expected utilities for Players A and B are defined as:
\(
u_A(\alpha_t, \beta_t) = \sum_{i=1}^N \sum_{j=1}^M \alpha_{t,i} \beta_{t,j} \mathbf{A}_{i,j} \quad \text{and} \quad u_B(\alpha_t, \beta_t) = \sum_{i=1}^N \sum_{j=1}^M \alpha_{t,i} \beta_{t,j} \mathbf{B}_{i,j} ,
\)
where \(\alpha_{t,i}\) and \(\beta_{t,j}\) represent the probabilities assigned to actions \(a_i\) and \(b_j\) at time \(t\), respectively.


In this work, we consider the interactions of two players, $A$ and $B$ who we refer to as the optimizer and learner respectively.

\textbf{Learner} (Player B):
The learner employs a learning algorithm that satisfies some \textit{no-regret} property. We consider two specific regret notions: external and swap-regret. We provide definitions in the appendix.
We note that no-regret algorithms exist and have been implemented for multiple online applications \citep{arora2005multiplicative}. No-swap-regret is a stronger condition, and algorithms that attain this property can be constructed from no-regret algorithms~\citep{cesa2006prediction}.

\textbf{Optimizer} (Player A): The optimizer aims to maximize their utility by best-responding to the learner's (Player B) adaptive behavior, without being constrained to a no-regret strategy. To facilitate this, the optimizer receives \textit{advice} about the learner's algorithm, providing insights into the learner's decision-making process.

A challenge the optimizer faces is how they can leverage advice to efficiently compute a strategy approximating the Stackelberg response. The Stackelberg strategy is the optimal response to a best responding opponent.The optimizer's Stackelberg response is utility-optimal against no-swap-regret learners and serves as a robust performance measure in applications like security games, though it may not be optimal against no-external-regret learners. Computing this strategy can be computationally expensive, especially against certain no-regret algorithms \citep{brown2023learninggamesgoodlearners}. The advice mitigates this complexity by offering a proxy for the learner's behavior, enabling efficient strategy computation.

We note that the optimizer's Stackelberg strategy is defined as:
\( \alpha^*:= \argmax_{\alpha \in \Delta(\mathcal{A})} u_A(\alpha, BR_B(\alpha)) \) whereby: \(
BR_B(\alpha) := \argmax_{\beta\in \Delta(\mathcal{B})}u_B(\alpha, \beta).\)

\subsection{Forms of Advice}


We consider two primary forms of advice that the optimizer may receive, each tailored to different game contexts:

\textbf{Simulator Advice:} A simulator $\mathcal{P}_B = \{ p_B^t: \mathcal{H}_t \rightarrow \Delta(\mathcal{B}) \}_{t=1}^T$, where $\mathcal{H}_t$ is the set of all possible histories up to time $t$. A history $H_t = (\alpha_1, \beta_1, \ldots, \alpha_t, \beta_t)$ comprises the mixed strategies played so far, and $p_B^t(H_t)$ predicts the learner's next mixed strategy $\beta_{t+1}$.

This form is valuable in complex games (e.g., large-scale spectrum auctions) where explicitly modeling the learner's payoffs is impractical. By simulating the learner's responses based on history, it aids the optimizer in approximating the Stackelberg strategy without requiring a full payoff matrix.

\textbf{Payoff Matrix Prediction:} A predicted payoff matrix \(\mathbf{B}_{{pred}}\) approximating the learner's true utility matrix \(\mathbf{B}\).
This advice allows the optimizer to compute a Stackelberg strategy as if \(\mathbf{B}_{{pred}}\) were accurate. It is natural in settings where the game's payoff structure can be explicitly represented, simplifying the optimization process.

In our exposition, we consider games for which there exists a unique Stackelberg strategy for the optimizer. Furthermore, we restrict ourselves to the subset of games for which it is feasible--under other learning algorithms--to efficiently determine an approximate Stackelberg strategy.  

\begin{assumption}
\label{BR_assumption}
Fix a game \( G = (\mathbf{A, B}) \) and value of $\varepsilon$,  we have that  
\begin{itemize}[nosep,leftmargin=*]
    \item There exists a unique Stackelberg strategy, $\alpha^\star \in \Delta(\mathcal{A})$. 
    \item $\mathbb{P}_{\alpha \sim \text{Unif}(\Delta(\mathcal{A}))} \left[ b_i \in BR_B(\alpha) \right] \in \{0\} \cup \left[1/\poly(\varepsilon^{-1}), 1\right]$. This is to say  the volume of each \emph{Best Response region} is either 0 or inverse polynomially large.
\end{itemize}
\end{assumption}


\vspace{-3pt}


\section{THEORETICAL RESULTS}
\label{sec: theory}
\vspace{-5pt}
In this section, we provide a framework for evaluating the quality of advice within repeated games. Evaluating the quality of advice in the repeated game setting presents unique challenges as not only does the type or form of advice matter, but also, the learning algorithm deployed by the opponent affects how useful any instance of advice is. 
\vspace{-3pt}
\subsection{Quality of Advice}
We establish a benchmark for advice mechanisms that can be applied to any opponent leveraging a learning algorithm that is either no-external-regret or no-swap-regret. To do this, we consider what the Stackelberg strategy induced by any advice instance is. This is to say, if we are to assume that the advice instance perfectly characterizes the opponent, what would the Stackelberg strategy be? We then define the pseudo-distance between any two advice instances to be the difference in utility, assuming the other player best responds, of the suggested Stackelberg strategies. This characterization allows us to develop a pseudo-metric over the space of advice for any particular advice instance. We can then assess how good any advice instance is by its distance, under the pseudo-metric, from the true realization of the learner. The elegance of this is that this metric generalizes to many other advice formats beyond the ones we concretely evaluate. Furthermore, making use of the Stackelberg strategy is well motivated as for opponents leveraging no-swap-regret algorithms, playing the Stackelberg strategy is utility-optimal up to $o(T)$ additive factors \citep{Deng2019StrategizingAN}. We describe the advice pseudo-metric and instantiate it within the two advice frameworks we provide above as an illustration.

\begin{definition}
    Fix a game $G = (\mathbf{A, B})$ and an advice instance ($\mathbf{B}_{pred}$ or $\mathcal{P}_B$). We define $\mathcal{S}(\mathbf{B}_{pred})$ or $\mathcal{S}(\mathcal{P}_B)$ as the Stackelberg strategy for the optimizer assuming that the advice ($\mathbf{B}_{pred}$ or $\mathcal{P}_B$) perfectly characterizes the learner.
\end{definition}
We now move on to describe the advice-pseudo distance, $\Gamma$, which will allow us to evaluate the quality of any advice. Here we define $\Gamma$ with respect to the predicted payoff matrix model of advice but note that this description extends to the simulator model of advice as well.
\begin{definition}
\label{def:advice_quality}
    Fix a game $G = (\mathbf{A}, \mathbf{B})$. Let $\mathbf{B}_{pred}$ and $\mathbf{B}'_{pred}$ be two advice instances, with $\mathcal{S}(\mathbf{B}_{pred})$ and  $\mathcal{S}(\mathbf{B}'_{pred}) \in \Delta(\mathcal{A})$ being the corresponding optimal Stackelberg strategies. We define the \emph{advice-pseudo-distance} between $\mathbf{B}_{pred}$ and $\mathbf{B}'_{pred}$ to be:
    \begin{align*}    \Gamma(\mathbf{B}_{pred}, \mathbf{B}'_{pred} ) := \big| &u_A(\mathcal{S}(\mathbf{B}_{pred}), BR_B(\mathcal{S}(\mathbf{B}_{pred}))) \\
    &- u_A(\mathcal{S}(\mathbf{B}'_{pred}), BR_B(\mathcal{S}(\mathbf{B}'_{pred}))) \big|.
    \end{align*}
\end{definition}

\begin{proposition}
    \label{prop:pseudo_metric}
    The advice-pseudo-distance $\Gamma$ for any game instance $G$ is a pseudo-metric.
\end{proposition}

We prove the above proposition in the appendix. The advice-pseudo-distance allows us to characterize what `good' advice is. For any advice instance, the advice-pseudo-distance with respect to the true realization of the learner defines how good that instance is. The larger the value of the advice-pseudo-distance, the worse the advice. The smaller the value of the advice-pseudo-distance, the better the advice.

\subsection{Efficient Stackelberg Computations via Good Advice}

In this section, we illustrate and motivate the use of advice in strategic settings by showing how having access to good simulators leads to more efficient computations of Stackelberg strategies \emph{even when considering the number of interactions with the simulator itself}. Previous work has shown that against some classes of no-regret algorithms (e.g., mean-based algorithms), it is necessarily the case that the optimizer would need to have an exponential number of interactions with the learner in order to determine an $\varepsilon$-approximate Stackelberg strategy \citep{brown2023learninggamesgoodlearners}. Good simulators can aid us in closing the gap in efficiency in determining Stackelberg strategies, serving as a good motivator for their use. 

We expound on important regret bound definitions here as they are useful in the exposition of the power of advice with respect to different types of learners. We note here that algorithms that satisfy no-anytime or no-adaptive regret bounds can be obtained from other base no-external-regret algorithms.

\begin{definition}[External and Swap Regret definition]
\begin{itemize}
    \item \textit{External Regret:} An algorithm is \textit{no-external-regret (no-regret)} if, for any sequence of actions by Player A, $(\alpha_1, \alpha_2, \ldots, \alpha_T)$, 
    \begin{align*}
    \sum_{t=1}^T u_B(\alpha_t, \beta_t) \geq \max_{b \in \mathcal{B}} \sum_{t=1}^T u_B(\alpha_t, b) - o(T).
    \end{align*}
    \item \textit{Swap Regret:} An algorithm is \textit{no-swap-regret} if, for any sequence of actions by Player A, $(\alpha_1, \alpha_2, \ldots, \alpha_T)$,
    \begin{align*}
        \sum_{t=1}^T u_B(\alpha_t, \beta_t) \geq \max_{\delta: \mathcal{B} \rightarrow \mathcal{B}} \sum_{t=1}^T u_B(\alpha_t, \delta(\beta_t)) - o(T),
    \end{align*}
    where $\delta:\mathcal{B}\rightarrow\mathcal{B}$ is a ``swap'' function mapping actions in $\mathcal{B}$ to itself. 
\end{itemize}
\end{definition}

\begin{definition}[$\phi$-anytime regret algorithms]
\label{anytime_regret_def}
An algorithm is an \emph{$\phi$-no-anytime-regret algorithm} if its regret satisfies
$
\text{Regret}(\tau) = O(\tau^c)
$
over the first \( \tau \) rounds, for some \( c < 1 \) and any \( \tau  \in [\phi, T] \). Here the regret is defined as
\[
\text{Regret}(\tau) \coloneqq  \max_{f \in \mathcal{F}} \sum_{t = 1}^{\tau} \left[ u_B(\alpha_{t, i}, f(\beta_{t, j})) - u_B(\alpha_{t,i}, \beta_{t,j}) \right].
\]
\end{definition}

\begin{definition}[$\phi$-adaptive regret algorithms]
\label{adaptive_regret_def}
An algorithm for player B is  \emph{$\phi$-no-adaptive-regret} if
\[
\sup_{r, s \in [T] : |s - r| \geq \phi} \text{Regret}([r, s]) \leq O(T^c), \quad \text{for some } c < 1,
\]
where
\begin{multline*}
\text{Regret}( [r, s]) \coloneqq \\
\max_{f \in \mathcal{F}}  \sum_{t = r}^{s} \left[ u_B(\alpha_{t,i}, f(\beta_{t,j})) - u_B(\alpha_{t,i}, \beta_{t,j}) \right].
\end{multline*}
\end{definition}

In Definitions \ref{adaptive_regret_def} and \ref{anytime_regret_def}, $\mathcal{F}$ can be the set of swap functions (in the case of swap-regret) or functions that map each action to a single, fixed action (in the case of external-regret).  The parameter $\phi$ can be a constant or a function of $T$ (e.g., $\log T$).


Previous work has identified a gap in the optimizer's ability to identify a Stackelberg strategy against different classes of no-regret learners. While with access to best-response oracles or against no-swap-regret learners the optimizer only needs a polynomial number of interactions to identify an $\varepsilon-$approximate Stackelberg equilibrium strategy, in certain games we see that against some no-anytime-regret learners (e.g., mean-based algorithms such as FTPL) the optimizer requires an exponential number of interactions.

\begin{proposition}[\citep{Letchford2009LearningAA}, \citep{PengLearningStrategiesCommit}]
For a game \( G \) satisfying Assumption \ref{BR_assumption}, there is an algorithm which finds an \( \varepsilon \)-approximate Stackelberg strategy for player \( A \) with  
\[
Q = \poly(|\mathcal{A}|, |\mathcal{B}|, 1/\varepsilon)
\]
queries to \( BR(\alpha) \).
\end{proposition}

\begin{proposition}
[\citep{brown2023learninggamesgoodlearners}]
     Consider the set $S$ of games satisfying Assumption \ref{BR_assumption}. There is an algorithm that finds  
an \( \varepsilon \)-approximate Stackelberg strategy in \( \poly(1/\varepsilon)^Q \) rounds against any no-anytime regret learner for any game in the set $S$. Here:
\[
Q = \poly(|\mathcal{A}|, |\mathcal{B}|, 1/\varepsilon).
\]
Furthermore, there exists a distribution over the set $S$ such that for a sampled Game $G$:
\begin{enumerate}
    \item For any no-swap-regret learner used by the opponent, there is a strategy for the learner which yields on average a reward for the optimizer that is $\varepsilon$ away from the Stackelberg value in $\poly(|\mathcal{A}|/ \varepsilon)$ interactions
    \item There exists a mean based-no-regret algorithm such that when used by the opponent, there is no strategy that yields on average a reward for the optimizer that is $\varepsilon$ away from the Stackelberg value unless they interact for $\exp(\Omega(|\mathcal{A}|))$ rounds with the learner.
\end{enumerate}
\end{proposition}

 Our first result shows that access to a good simulator closes this gap. In particular, if the simulator of the anytime no-regret learner is able to yield an approximate Stackelberg strategy, the advice-augmented optimizer is also able to identify an approximate strategy with a polynomial number of interactions even when considering the number of interactions they have with the simulator.

\begin{assumption}
\label{strength_pred}
    For a game $G = (\mathbf{A, B})$ satisfying Assumption \ref{BR_assumption}. Let $\mathcal{P}_B$ be such that:
    \begin{enumerate} 
        \item $\mathcal{P}_B$ is $\poly(1/ \varepsilon)$-no-anytime-regret.
        \item $\Gamma(\mathcal{P}_B, \mathbf{B}) \leq \delta$.
    \end{enumerate}
\end{assumption}

\begin{theorem}
\label{thm:efficient_computation}
    Let $\mathcal{P_B}$ be a prediction of the learning algorithm being utilized by player $B$ that satisfies Assumption \ref{strength_pred}. We then have that in $\poly(|\mathcal{A}|, |\mathcal{B}|, 1/\varepsilon)$  number of interactions with both the predictor and player $B$, player $A$ is able to find a strategy which yields on average a reward for the optimizer that is $\delta + \varepsilon$ away from the Stackelberg value. 
\end{theorem}

The above result, while interesting in that it shows the power of having a simulator, has some strong conditions on the quality of the advice. The condition of being no-regret, coupled with the error bounds, implies that the simulator in the long run is going to be accurately best responding to most queries. The assumption of being $\poly(1/ \varepsilon)$-anytime-no-regret, on the other hand, is an assumption that is easily satisfied by a wide range of no-regret learning algorithms including mean-based algorithms. Therefore, the result above demonstrates that with a sufficiently `good' or `accurate' simulator the optimizer can require fewer interactions to identify a Stackelberg strategy.

\subsection{Consistent No-Regret Advice-Augmented Algorithms}
Moving beyond the identification of a Stackelberg strategy, we consider how the optimizer may leverage advice when provided no guarantees of correctness. A major motivation in the development of advice-augmented algorithms is to design algorithms that retain a notion of `robustness' yet can take advantage of advice. In other words, these algorithms have the `best of both worlds' in the sense that they are able to be robust in instances when the advice is incorrect yet in the cases that the advice is correct, take advantage of the benefit of having the advice. We investigate the extent to which this is possible in repeated games. Most robust approaches in the repeated game setting have focused on developing learning algorithms that have some `no-regret' notion. We can thus think of `no-regret' notions as forms of robustness. 
\subsubsection{Motivating Example: Advice in 2 $\times$2 Games}
Before we present general results for leveraging advice in the repeated game setting we provide motivation from a repeated $2 \times 2$ game interaction. In this particular setting, we show how the optimizer can leverage an advice-augmented algorithm to ensure that they are close to the Stackelberg equilibrium in instances that the advice is close to the truth yet retain a no-regret property in instances the advice is not. Consider the following game $G$ as specified by the following payoff matrices:
\[
\begin{array}{ccc}
\text{Player A} & \text{Player B} & \text{Advice} \\
\mathbf{A} = \begin{bmatrix}
0 & 0 \\
1 & -10
\end{bmatrix}
&
\mathbf{B} = \begin{bmatrix}
3 & -5 \\
1 & -1
\end{bmatrix}
&
\mathbf{B}_{pred} = \begin{bmatrix}
2.5 & -4 \\
1 & -0.5
\end{bmatrix}
\end{array}
\]
In this game, the Stackelberg strategy for Player $A$ is $
\begin{bmatrix}
\frac{2}{3} \\
\frac{1}{3}
\end{bmatrix}
$, which induces a best response of $\begin{bmatrix}
1 \\
0
\end{bmatrix}$
 by the learner. The prediction $\mathbf{B}_{pred}$ results in a predicted Stackelberg strategy of $
\begin{bmatrix}
0.7 \\
0.3
\end{bmatrix}$. Note that finding the Stackelberg strategy for Player $A$ in this game reduces to finding a value of $p \in [0, 1]$ as the optimizer's Stackelberg problem is defined over the one dimensional probability simplex. It indeed is possible to design advice-augmented algorithms that achieve a `best of both worlds' guarantee for this game and more generally, repeated $2 \times 2$ games against no-regret learners. In particular, the optimizer can implement a verify and check threshold based approach to decide whether they should leverage the advice or not. Concretely, the optimizer could check if $\begin{bmatrix}
0.7 + \frac{1}{2}\delta \\
0.3 - \frac{1}{2}\delta
\end{bmatrix}$ and $\begin{bmatrix}
0.7 - \frac{1}{2}\delta\\
0.3 + \frac{1}{2}\delta
\end{bmatrix}$ result in different best responses from a no-regret learner. This would mean that the Stackelberg strategy is at most $\delta$ away from the advice's predicted Stackelberg strategy. This therefore allows them to decide if the advice's Stackelberg strategy is sufficiently accurate. We formalize this example through the following proposition.
 \begin{definition}
    We define a $\delta$-approximate Stackelberg strategy
    \[ \alpha^*_{\delta} \in \argmax_{\alpha \in \Delta(\mathcal{A})}u_A(\alpha, BR_B(\alpha)) \text{ such that } \|\alpha^\star- \alpha\|_1 \leq \delta, \] where $\alpha^\star$ is the true Stackelberg strategy. We denote the corresponding utility at $\alpha^*_{\delta}$ the $\delta$-approximate Stackelberg value $\kappa(\delta)$.
 \end{definition}
 \begin{proposition}
 \label{proposition_2x2}
      Let $V$ be the true Stackelberg value and $\mathbf{B}_{pred}$ be a prediction of the payoff matrix $\mathbf{B}$ in a repeated $2 \times 2$ game. For any fixed threshold value of $\delta$, there exists an advice-augmented algorithm $ALG$ such that:
     \begin{itemize}
         \item If $\Gamma(\mathbf{B, B}_{pred}) < V - \kappa(\delta)$ then $\sum\limits_{t = 1}^{T} u_A(\alpha_t, \beta_t) \geq  \kappa(\delta)T - o(T)$;
         \item if $\Gamma(\mathbf{B, B}_{pred}) \geq V - \kappa(\delta)$ then  ALG is no-regret,
     \end{itemize}
     against a $\phi$-no-adaptive regret learner when $\phi \in o(T)$
 \end{proposition}
 
 \subsubsection{Impossibility of General Robust and Consistent Advice-Augmented Algorithms}
 While the previous result showcases some potential promise when it comes to designing advice-augmented algorithms in repeated games, we ask whether advice-augmented algorithms with best of both world guarantees can extend to more general settings. We therefore extend our investigation to the case of general no-external-regret / no-swap-regret learners and consider games beyond the $2 \times 2$ setting. We find that in the general case, an optimizer can not guarantee that for all instances of games and all no-regret learners, they will be able to make use of an advice-augmented algorithm that achieves an approximation of the Stackelberg value when the advice is approximately accurate and guarantee no-regret robustness when the advice is inaccurate. 
\begin{theorem}
\label{thm:impossibility}
    There exists a repeated game $G$ which satisfies Assumption \ref{BR_assumption} such that there does NOT exist an advice-augmented algorithm $ALG$ for which for any fixed threshold value of $\delta$:
\begin{itemize}
    \item If $\Gamma(\mathbf{B, B}_{pred}) < V - \kappa(\delta)$ then $\sum\limits_{t = 1}^{T} u_A(\alpha_t, \beta_t) \geq  \kappa(\delta)T - o(T)$;
    \item if $\Gamma(\mathbf{B, B}_{pred}) \geq V - \kappa(\delta)$ then  ALG is no-regret,
     \end{itemize}
against any no-regret or no-swap-regret learner.
\end{theorem}

This main theorem illustrates how it is impossible to switch between retaining a no-regret property in cases where the advice is incorrect and attaining utility close to the Stackelberg value in cases where the advice is correct. This, strikingly, shows the difficulty of appropriately leveraging advice in the repeated game setting. Driving the impossibility result above is the wide range of no-regret adversaries that exist. In particular, we consider a no-regret adversary that, to whichever no-regret advice-augmented algorithm, steers the no-regret algorithm away from actions whose best response for the learner is a Stackelberg follower strategy.

The result above holds for $\phi = T$ meaning that we are considering the space of very strong adversaries who only have to satisfy a no-regret guarantee over the entire time horizon. In many practical settings, however, strategic opponents aim to have a no-regret guarantee in an adaptive sense. This is to say over windows of time, they hope to be satisfying some no-regret condition. Furthermore, a rough estimate of the extent to which they are generally able to satisfy this condition may be common knowledge to all players. A concrete example where this assumption is plausible (and for which we have a numerical illustration) is the security game (see Section \ref{sec: numerics}). In these instances, players would want to guarantee some no-regret condition even over windows of time. Furthermore, the amount of regret they are willing to incur may be deduced by opponents. As such we consider the setting where the $\phi$ is known by both the learner and the optimizer. Furthermore, let $\varepsilon$ be the maximum regret in any window the learner can incur i.e.,
\(
\sup_{r, s \in [T] : |s - r| \geq \phi} \text{Regret}([r, s]) \leq \varepsilon
\). We assume that both the optimizer and learner are aware of this as well. 

We begin by attempting to prove a version of Theorem \ref{thm:impossibility} with $\phi$-adaptive adversaries whose regret is bounded by $\varepsilon$ in any window. With this relaxation, we ask whether having access to advice allows us to derive any benefit in the case that the advice is correct whilst guaranteeing some no-regret property in the cases that the advice is incorrect. For fixed values of $|\mathcal{A}|, |\mathcal{B}|$, as $T$ grows to infinity this problem in general becomes trivial as the optimizer can simply make use of Stackelberg identification procedures to find the optimal Stackelberg strategy. Thus, we are interested in cases where $T = o(|\mathcal{A}||\mathcal{B}|)$. In these cases, the optimizer cannot make use of Stackelberg identifying procedures. What we find, interestingly, is that guaranteeing the correctness of a prediction of a Stackelberg strategy ends up being the same as finding a Stackelberg strategy in the first place. This seemingly renders the advice redundant in this particular setting further illustrating the robustness of  Theorem \ref{thm:impossibility}. 

\begin{proposition}
\label{prop:same_count}
  Fix a value of $\phi$ and $\delta$. There exists a game $G$ satisfying ~\Cref{BR_assumption} such that any advice-augmented optimizer needs $\Omega(\phi |\mathcal{A}||\mathcal{B}|)$ interactions with a $\phi$-no-adaptive regret learner to confirm if some advice instance yields a $\delta$-approximate Stackelberg strategy.  
\end{proposition}

The main intuition behind this proposition is that when provided with advice that you can not guarantee its correctness, in some games the best option for the optimizer is to simply engage with the learner in the same way they would have without the advice in order to establish its correctness. 

Motivated by the impossibility results above, we shift our perspective and relax the optimizer's goal when leveraging advice-augmented algorithms. Instead of trying to guarantee near Stackelberg performance, we relax this and ask the question of whether they can simply benefit from the advice when it is better than what they would have gotten from a regret minimizing trajectory. The theorem above showed that there does not exist an advice-augmented algorithm that, for all instances of no-regret adversaries, could guarantee to both be no-regret when the advice is below some accuracy threshold and be close to the Stackelberg value when the advice is above an accuracy threshold. We now instead ask if it is possible to design an advice-augmented algorithm that deploys a strategy that weakly dominates some (coarse)-correlated equilibrium. We note that convergence to (coarse)-correlated equilibrium is a consequence of jointly leveraging no-(swap)-regret algorithms.

In this setting, we assume that the optimizer knows values of $\phi$ and $\varepsilon$ as described above. We note here that they need not know these values exactly for the algorithm to work. It is only important that they are aware and make use of values of $\phi$ and $\varepsilon$ that are at least as large as the true values the learner utilizes.
\begin{theorem}
\label{thm:no-regret}
    There exists an advice-augmented algorithm for the optimizer such that when they engage with a $\phi$-no-adaptive-external regret learner incurring at most $\varepsilon$ regret in each $\phi$ window, the optimizer's total expected utility:
    \begin{enumerate}[label=(\roman*)]
        \item weakly dominates their utility from some $\varepsilon$-approximate coarse correlated equilibrium, and
        \item strictly dominates the utility from a coarse correlated equilibrium for some $\phi$-no-adaptive-external regret learner and advice instances.
    \end{enumerate}
\end{theorem}

Theorem \ref{thm:no-regret} moves to establish both a version of robustness and consistency. In particular, we have the weak-domination of a (C)CE regardless of the exact form of $\phi$-no-adaptive-external regret learner they encounter. However, the theorem also establishes the possibility of strict domination of a (C)CE  illustrating the potential for leveraging the benefits of having access to good advice. Variations of common no-regret algorithms satisfy the $\phi-$no-adaptive-external regret criterion(e.g., variants of multiplicative weights)\citep{pmlr-v80-zhang18o, hazan_odp}. These would thus be instances of learners for which good advice results in the strict domination of (C)CE. 

We note that while Theorem \ref{thm:no-regret} is written for the case of $\phi$-no-adaptive-external-regret players and coarse-correlated equilibria, we can easily extend the same idea to no-swap-regret players and correlated equilibria (see~\Cref{cor:no-regret-ext}). The main idea behind the proofs is that the optimizer can make use of a `trying' and `checking' approach in order to determine whether or not to make use of the advice. 

We also note that in the proof of ~\Cref{thm:no-regret}, the complexity associated with diverging trajectories of play from having tried to play the advice means that the optimizer has to take into account the possibility of shifts in equilibria when changing back to a regret minimizing approach. 

\vspace{-5pt}

\section{NUMERICAL RESULTS}
\label{sec: numerics}
One particularly important real world application in which repeated interactions and Stackelberg strategies in particular arise is the domain of security games. In this setting two players, an attacker and defender, repeatedly interact over time and it often is the case that the strategy the defender seeks to implement is the Stackelberg strategy. For both the attacker and the defender, uncertainty of the other's true risk appetite and resources often leads to them not having a concrete understanding of the opponent's payoff matrix. Increasingly machine learning is being leveraged in these domains in order to guide defenders in the identification of optimal strategies.

\begin{figure}[htbp]
  \centering
  \begin{subfigure}[t]{0.45\textwidth}
    \centering
    \includegraphics[width=\textwidth]{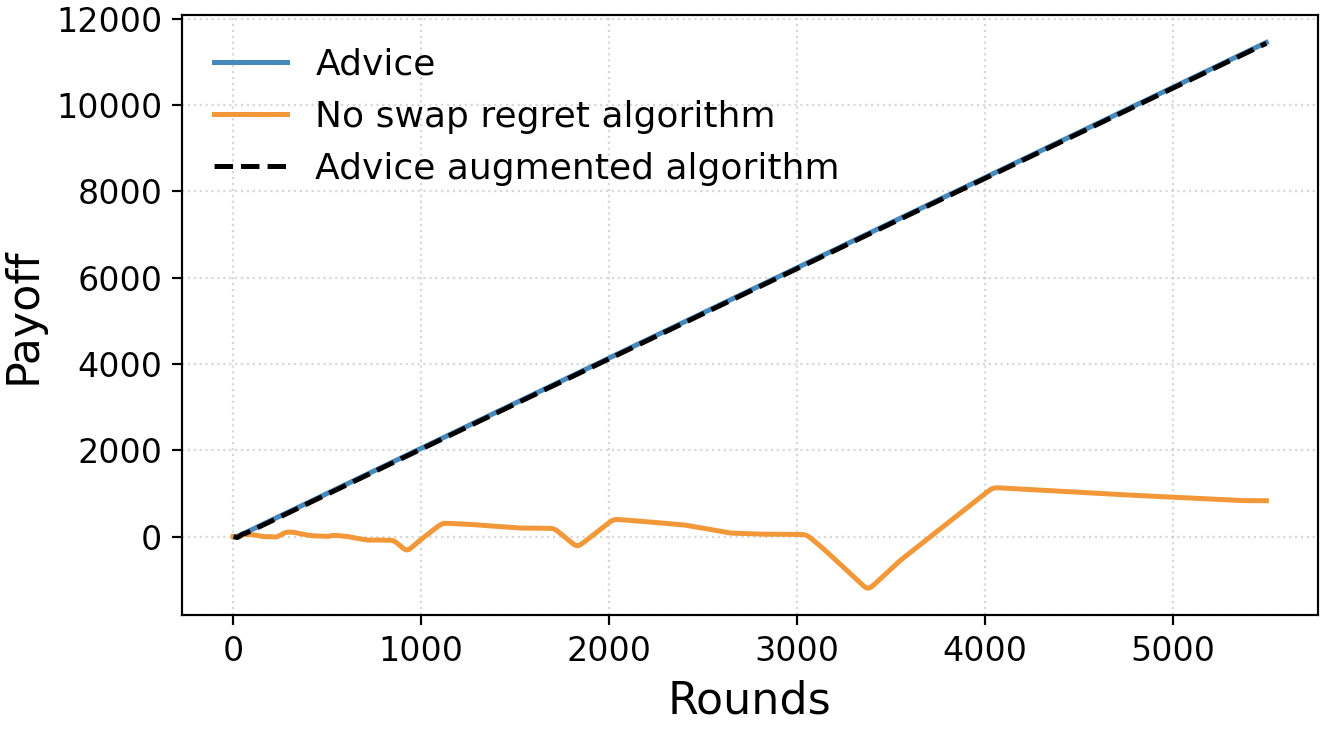}
    \caption{Optimizer cumulative payoff given true payoff matrix as advice}
  \end{subfigure}
  \hfill
  \begin{subfigure}[t]{0.45\textwidth}
    \centering
    \includegraphics[width=\textwidth]{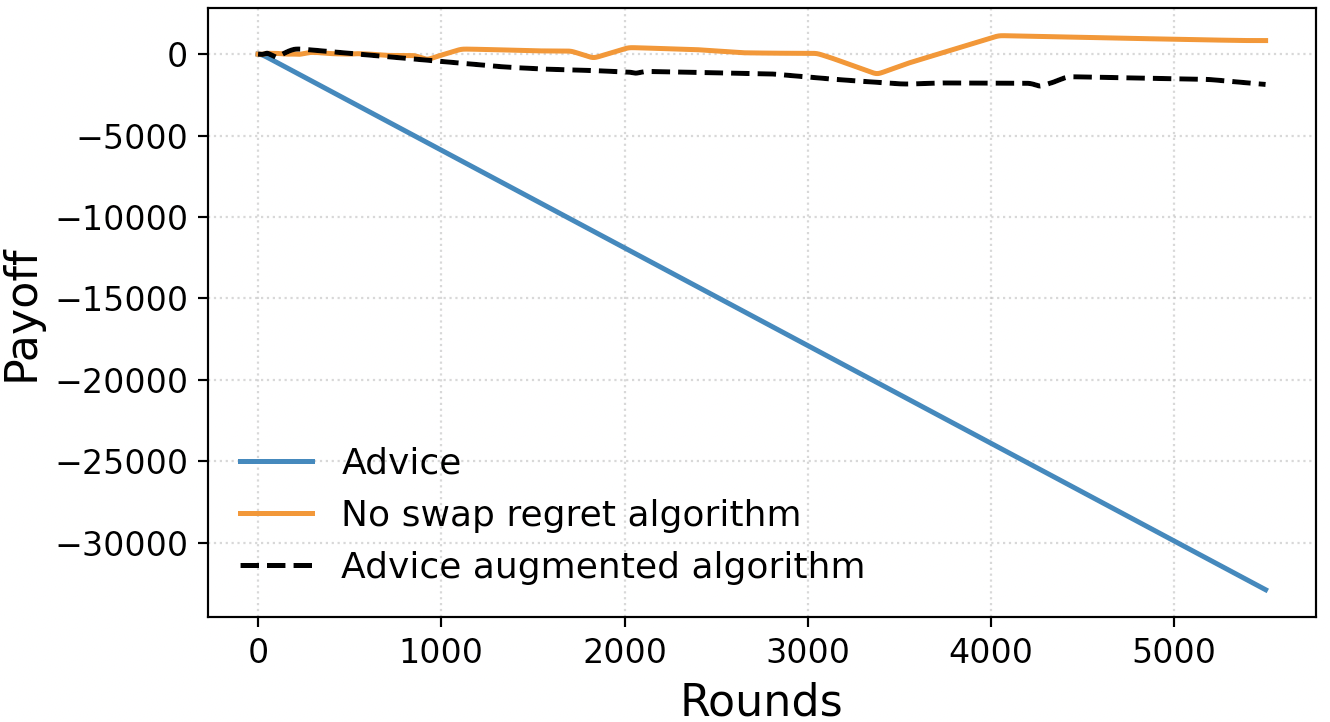}
    \caption{Optimizer cumulative payoff given incorrect payoff matrix prediction as advice}
  \end{subfigure}
  \caption{Comparison of optimizer cumulative payoff under different advice instances.}
  \label{fig:comparison-grid}
\end{figure}

We provide an initial empirical evaluation of advice-augmented algorithms in the security game setting with a focus on the \emph{green security game}. Building on work such as \citep{fang2015when}, we leverage a model of interaction between park rangers who are defenders and illegal poachers in a national park. In this scenario, the park rangers (the optimizer) would like to catch illegal poachers (the learner). The poachers would like to find animals (in our case African elephants) without being detected by the park rangers. The park rangers are provided advice in the form of payoff matrix predictions. We provide more details on the setup, explicit strategy constructions and additional evaluations in Appendix \ref{Exp_app}.



Through our experiments, we observe the behavior we hope for. In particular, if the advice is good, the algorithm is able to take advantage of the prediction and get better performance and in the instance that the advice is not as good, the algorithm is able to protect itself and default to a safer option with regret guarantees. It is of note that we show as a benchmark, the trajectories that a no-swap-regret algorithm would have taken if it had only been running as a no-swap-regret algorithm. We also show the cumulative payoff that the advice would have gotten if it had been the only strategy played from the beginning.

\section{CONCLUSION}
\label{sec: conclusion}
This work provides an investigation into the incorporation of machine-learned advice within the repeated game context. We characterize the advice landscape and provide a concrete pseudo-metric to aid in understanding what `good' advice could look like. We illustrate the power that good forms of advice could yield in determining Stackelberg strategies. We show the limitations of advice. We prove that it is impossible to, in general, guarantee both near Stackelberg performance when the advice is correct and some no-regret condition in the cases that the advice is incorrect. Despite this, we find a more natural fallback wherein against certain classes of learners, the optimizer is able to deploy an advice-augmented strategy that weakly dominates the utility from some (coarse)-correlated equilibrium. This, in effect, allows for the optimizer to at least take advantage of the advice in some cases where it is indeed possible for them to do so, whilst guaranteeing a no-regret condition in instances where that is not possible. We provide a numerical illustration from real word data which we cast as a green-security game. Altogether our findings survey the landscape of incorporating machine-learned advice within the repeated game setting. 
\newpage

\bibliographystyle{plainnat}  
\bibliography{NeurIPS} 

\newpage
\appendix
\label{appendix}
\newpage
\section*{Appendix}
 \textbf{Limitations and Future Directions} While our work considers the problem only from the optimizer's point of view, we plan to conduct additional analysis looking into the impact of having multiple agents all leveraging advice augmented algorithms. Furthermore, in this work, we make use of the Stackelberg strategy as a desirable goal for the optimizer. While this indeed draws justification from both a practical perspective (i.e., security games) and a theoretical perspective (i.e., $o(T)$ additive factor optimality when faced with a no-swap-regret opponent), there is an opportunity to consider advice benchmarked on optimality notions closer to absolute optimality over all possible trajectories for a given opponent.

\textbf{Broader Impacts} Our work has potential impact in a wide range of fields where strategic decisions are made by players over time in a repeated manner. These domains include sectors such as security, finance, etc., Our results enable the development of safe robust algorithms in these settings particularly as machine learning becomes a bigger part of how decisions are made in these settings. There are however potential negative impacts specifically arising from the improper development and misuse of prediction models and as a consequence the algorithms. Future research could therefore be aimed at understanding how to mitigate such risks. 

\section{Proof details}

\subsection{Proof of Proposition \ref{prop:pseudo_metric}}
\begin{proof}
    We proceed to show this by verification of the properties of a pseudo-metric:
    \begin{itemize}[nosep,leftmargin=*]
        \item For any $\mathbf{B}_{pred}$ we have that $\Gamma(\mathbf{B}_{pred}, \mathbf{B}_{pred}) = | u_A(\mathcal{S}(\mathbf{B}_{pred}), BR_B(\mathcal{S}(\mathbf{B}_{pred}))) - u_A(\mathcal{S}(\mathbf{B}_{pred}), BR_B(\mathcal{S}(\mathbf{B}_{pred})))| = 0$
        \item For any $\mathbf{B}_{pred}, \mathbf{B}'_{pred}$ we have that $\Gamma(\mathbf{B}_{pred},\mathbf{B}'_{pred}) = | u_A(\mathcal{S}(\mathbf{B}_{pred}), BR_B(\mathcal{S}(\mathbf{B}_{pred}))) - u_A(\mathcal{S}(\mathbf{B}'_{pred}), BR_B(\mathcal{S}(\mathbf{B}'_{pred}))) | = |  u_A(\mathcal{S}(\mathbf{B}'_{pred}), BR_B(\mathcal{S}(\mathbf{B}'_{pred}))) - u_A(\mathcal{S}(\mathbf{B}_{pred}), BR_B(\mathcal{S}(\mathbf{B}_{pred}))) | = \Gamma(\mathbf{B}'_{pred},\mathbf{B}_{pred}) $. Thus illustrating symmetricity.
        \item For any $\mathbf{B}'_{pred}, \mathbf{B}_{pred}, \bar{\mathbf{B}}_{pred}$ consider $u_A(\mathcal{S}(\mathbf{B}_{pred}), BR_B(\mathcal{S}(\mathbf{B}_{pred}))),$ \\ $ u_A(\mathcal{S}(\mathbf{B}'_{pred}), BR_B(\mathcal{S}(\mathbf{B}'_{pred}))),$ and $ u_A(\mathcal{S}(\bar{\mathbf{B}}_{pred}), BR_B(\mathcal{S}(\bar{\mathbf{B}}_{pred})))$. Note that these are three points on the real line and from the satisfaction of the triangle inequality on the real line, we can then derive the triangle inequality property for $\Gamma$ for any instances of $\mathbf{B}'_{pred}, \mathbf{B}_{pred}, \bar{\mathbf{B}}_{pred}$ . 
    \end{itemize}
\end{proof}

\subsection{Proof of Proposition \ref{proposition_2x2}}
\begin{proof}
     Let $\mathcal{S}(\mathbf{B}) = \begin{bmatrix}
p^\star \\
1-p^\star
\end{bmatrix} $ and  $\mathcal{S}(\mathbf{B}_{pred}) = \begin{bmatrix}
p \\
1-p
\end{bmatrix}$. Note that $ \| \begin{bmatrix}
p^\star \\
1-p^\star
\end{bmatrix} - \begin{bmatrix}
p \\
1-p
\end{bmatrix}\|_1 = 2|p - p^\star|$. The condition $\Gamma(\mathbf{B, B}_{pred}) < V - \kappa(\delta)$ thus implies that $|p - p^\star| < \frac{1}{2}\delta$. Without loss of generality, assume that $[p^\star - \frac{1}{2}\delta , p^\star +\frac{1}{2}\delta] \subset (0, 1)$. Let $R_1 = \{p \in [0,1]: BR_B( \begin{bmatrix}
p \\
1-p
\end{bmatrix}) =  \begin{bmatrix}
0 \\
1
\end{bmatrix} \}$ and $R_2 = \{p \in [0,1]: BR_B( \begin{bmatrix}
p \\
1-p
\end{bmatrix}) =  \begin{bmatrix}
1 \\
0
\end{bmatrix} \}$. Since the Stackelberg strategy is $\begin{bmatrix}
p^\star \\
1-p^\star
\end{bmatrix}$, realize that $p^\star + \epsilon$ and $p^\star - \epsilon$ are in different regions $\forall \epsilon$ (i.e., if $p^\star + \epsilon \in R_1$ then $p^\star - \epsilon \in R_2$). We can then determine if $|p - p^\star| \leq \frac{1}{2}\delta$ by checking if $p + \frac{1}{2}\delta$ and  $p - \frac{1}{2}\delta$ are in the same region. \\

Since the learner is $\phi$-no-adaptive regret learner with $\phi \in o(T)$, we note that the optimizer can employ queries by playing the same strategy for an $o(T)$ length of time. The $\phi$-no-adaptive regret condition ensures convergence to a best response strategy within $o(T)$ rounds. These two facts provide a nice characterization of an advice augmented algorithm for the optimizer. The optimizer simply queries $\begin{bmatrix}
p + \frac{1}{2}\delta \\
1-p - \frac{1}{2}\delta
\end{bmatrix}$, $\begin{bmatrix}
p \\
1-p
\end{bmatrix}$   and $\begin{bmatrix}
p - \frac{1}{2}\delta \\
1-p + \frac{1}{2}\delta
\end{bmatrix}$. If the best responses to these queries are different, the optimizer then plays the strategy with the most utility else they play a regular no-regret strategy (e.g., Multiplicative Weights).

In the case that $\Gamma(\mathbf{B, B}_{pred}) < V - \kappa(\delta)$, the optimizer does not attain the approximate Stackelberg value for at most 4$\times o(T)$ rounds from three query rounds and an additional round from having selected the strategy to play but waiting for the algorithm to best respond. If the advice is beyond the accuracy threshold, the optimizer simply incurs an additional $3 \times o(T)$ regret from the `queries' used in the verification part of the algorithm. This means $ALG$ would still retain the no-regret property of the out of the box no-regret algorithm it leverages.
 \end{proof}

\subsection{Proof of Proposition \ref{prop:same_count}}
\begin{proof}

Let the best response regions, $R_b$, be sets in $\Delta(\mathcal{A})$ such that $\forall \alpha \in R_b, b = BR_B(\alpha)$. We denote the collection of these sets to be $\mathbf{R}$. Furthermore, we note that these regions are polytopes separated by hyperplanes. 

For some games, it is the case that to determine the location of the hyperplanes separating the best response regions, the optimizer needs $\Omega(\phi |\mathcal{A}||\mathcal{B}|)$ interactions. In particular consider the set of games $S$ where $|\mathbf{R}| = |\mathcal{B}|$. We get this from $\phi$ interactions to get a best response for each query from the learner, $\Omega(|\mathcal{B}|)$ hyperplanes to separate each of the best response regions, and then $\Omega(|\mathcal{A}|)$ to determine each hyperplane in the $|\mathcal{A}|-1$ dimensional simplex. 

Assume toward a contradiction that an advice-augmented algorithm $ALG$ confirmed some strategy $\bar{\alpha}$ as being a $\delta-$approximate Stackelberg strategy without determining all approximate hyperplanes. We can easily construct another game $G'$ such that the Stackelberg strategy lies on a hyperplane $H$ which is undetermined by $ALG$. This then makes it possible to have $G'$ such that the strategy $\bar{\alpha}$
is not a $\delta-$approximate Stackelberg strategy. For $G'$, given the same advice and the same interactions, this confirmation of $\bar{a}$ is incorrect.
\end{proof}

\subsection{Proof of Theorem \ref{thm:efficient_computation}}
\begin{proof}
    We note that the optimizer can always generate histories that query the simulator's best response for a particular fixed strategy. This makes the simulator, in essence, an approximation of a best response oracle. To put this concretely, to implement a query, the optimizer generates a transcript $\mathcal{O}(\poly(1/\varepsilon))$ plays of a fixed strategy $q$. The no-anytime-regret guarantee of $\mathcal{P}_B$ implies that the optimizer gets the fixed best response to the query $q$. They can then repeat this, starting from the first action. This is possible as the optimizer has access to the simulator's functions and can always start from the beginning again. This thus fashions $\mathcal{P}_B$ to be a form of a best response oracle. The no-anytime-regret property combined with the performance guarantee on the simulator means that they can identify an $\varepsilon-$approximation of the $\mathcal{P}_B$'s Stackelberg strategy. Accounting for the potential error in $\mathcal{P}_B$, we thus have a $\delta + \varepsilon-$approximate Stackelberg strategy. We have from \cite{Letchford2009LearningAA} that the optimizer needs $Q = \poly(|\mathcal{A}|, |\mathcal{B}|, 1/\varepsilon)$ queries which then gives us the total number of queries required $Q \cdot \poly(1/\varepsilon) = \poly(|\mathcal{A}|, |\mathcal{B}|, 1/\varepsilon)$.
\end{proof}

\subsection{Proof of Theorem \ref{thm:impossibility}}
\begin{proof}
    The main idea behind this proof is that even with correct advice, a no-regret adversary can pretend to have a completely different payoff matrix and induce the optimizer to abandon playing a (approximate) Stackelberg strategy. Throughout this proof, we assume that the learner maintains a no-swap-regret condition. We note that the same analysis will go through for a no-external-regret learner. 
    
    We create an adversarial instance of player $B$ whose goal is to convince player $A$ not to play an approximate Stackelberg equilibrium. Let the Stackelberg equilibrium be $(\alpha^\star, b^\star)$ for some fixed $b^\star \in \mathcal{B}$. Let $R^\star := \{\alpha \in \Delta({\mathcal{A}}) \text{ such that } BR_B(\alpha) = b^\star \}$. $R^\star$ is essentially the facet that induces player $B$ to play the action $b^\star$ as a best response. To convince the player not to play an approximate Stackelberg strategy, Player $B$ will `punish' player $A$ whenever they play a strategy that would force them to best respond with $b^\star$. Player $B$, however, otherwise plays some regular no-swap-regret strategy in order to maintain a no-swap-regret property.
    
    To implement this strategy Player $B$ makes use of the set $\{ \alpha_{\tau}: \alpha_\tau \in R^\star , \tau < t\}$. They keep track of the size of this set, and if at any point it exceeds a threshold, they play some action that forces player $A$ to enact strategies outside of region $R^\star$. Player $B$ makes sure that the Player $A$ has played strategies in $R^\star$ sufficiently many times to ensure that they are aware of convergence to strategies within this space and that they meaningfully incur regret for some convergence into the strategy space $R^\star$.
    
    We make this concrete with the following algorithmic description. Assume player $B$ has access to a no-swap-regret algorithm $B_{NSR}$. The player then plays the following strategy at each time $t$:

\begin{algorithm}[!ht]
    \KwIn{Sequence $\{ (\alpha_\tau, \beta_\tau) \}_{\tau = 1}^{t-1}$}
    \If(\tcp*[f]{for some $g(T) \in o(T) \cap \omega(\log T)$}){$c \leq \left| \left\{ \alpha_{\tau} : \alpha_\tau \in R^\star ,\ \tau < t \right\} \right| < g(T)$}{
        $\beta_t \gets b'$ \; 
    }
    \Else{
        $\beta_t \gets B_{NSR}(\{ (\alpha_\tau, \beta_\tau) \}_{\tau = 1}^{t-1})$ \;
    }
    \caption{Player $B$'s Decision Rule for $\beta_t$}
    \label{Alg_ref1}
\end{algorithm}

We show that for any optimizer (player A) that attempts to retain a no-regret guarantee in the case of the advice being incorrect, there exists an instance of Algorithm \ref{Alg_ref1} which forces the optimizer to always be no-regret even if the advice is perfectly correct. As such assume $\Gamma(\mathbf{B}_{pred}, \mathbf{B}) = 0$. We consider a game such that $BR_A(b') \notin R^\star$. Firstly, we note that the Algorithm \ref{Alg_ref1} for player $B$ is no-swap-regret. In the worst case, it adds an additional $o(T)$ factor of regret. This coupled with the underlying no-swap-regret guarantee of $B_{NSR}$ means that the algorithm is inherits a no-swap-regret guarantee. 
    
    To illustrate why the optimizer can be forced to always be no-regret, consider the optimizer's problem. The optimizer has the challenge of discerning whether the responses they get from Player $B$ at any instant indeed are accurate and reflect some true payoff matrix. In the context of Algorithm \ref{Alg_ref1}, the optimizer has to decide whether when they see player $B$ play the strategy $b'$ if it indeed is a best response strategy. We argue that there will always exist some no-swap-regret learner who can continue playing the strategy $b'$ for long enough that Player $A$ defaults to simply playing a no-swap-regret strategy themselves. 
    
    Concretely, consider some other version of a payoff matrix for the learner, $\mathbf{B}'$, such that $\forall \alpha \in R^\star, BR_{\mathbf{B}'}(\alpha) = b'$ with $BR_A(b') \notin R^\star$ (we note that $BR_{\mathbf{B}'}$ is a slight abuse of notation to signify the best response according to Player $B$ assuming their payoff matrix is $\mathbf{B}'$). This is to say that it indeed is the case that the best response to the optimizer playing a strategy in $R^\star$ is $b'$. The optimizer has to in $o(T)$ time either commit to the advice's proposed strategy without full knowledge of whether the learner's payoff matrix is $\mathbf{B}'$ or $\mathbf{B}_{pred}$ or default to a no-swap-regret strategy and guarantee no-swap-regret. Realize that as a consequence we have the magnitude of Player $A$'s historic plays in region $R^\star$ as being either $\in o(T)$ or $\notin o(T)$.

    \textbf{Case 1}: $ |\{ \alpha_{\tau} : \alpha_\tau \in R^\star , \tau < T\}| \in o(T\})$:\\
    In this instance, the optimizer straightforwardly violates the first condition as they do not take advantage of the advice despite it being accurate.

    \textbf{Case 2}: $ |\{ \alpha_{\tau} : \alpha_\tau \in R^\star , \tau < T\}| \notin o(T)\}$:\\
    This would mean that for a version of the game where $\mathbf{B}'$ was the true payoff matrix for the learner, the optimizer would not retain the no-regret property.
\end{proof}
    
\subsection{Proof of Theorem \ref{thm:no-regret}}
\begin{proof}
Consider Algorithm \ref{Alg_nr_dominate}. We assume that the optimizer also has access to a $\phi$-no-adaptive-external regret learner incurring at most $\varepsilon$ regret in each $\phi$ window, $A_{NR}$. There are two cases to analyze here:

\textbf{Case 1:} Algorithm \ref{Alg_nr_dominate} in Phase $4$ commits to $\alpha^*_{advice}$:\\
Note that this only happens after a phase in which the total utility induced by $\alpha^*_{advice}$ is larger than one induced by playing $A_{NR}$. Note that since both players are playing no-regret strategies guaranteed to have at most regret $\varepsilon$, we get convergence to $\varepsilon$-approximate coarse correlated equilibria. As such we have that the utility we are getting from $\alpha^*_{advice}$ dominates what we would have gotten from a $\varepsilon$-approximate coarse correlated equilibrium, in particular, the one the optimizer compared against. 
    
\textbf{Case 2:} Algorithm \ref{Alg_nr_dominate} in Phase $4$ commits to $A_{NR}$:\\
This happens after having played $A_{NR}$ and having witnessed that the $\varepsilon-$approximate coarse correlated equilibrium has higher utility compared to the advice. The confirmation step here is essential as after having played some other strategy for a while, the optimizer in Phase 3 confirms that they are able to get on a trajectory that still is better than having followed the advice. If having played the advice previously alters the trajectory and they find themselves in some other coarse correlated equilibrium that is worse than the utility they get from the advice they can still opt out. However, for them to commit to $A_{NR}$, it must be the case that the utility they have received from the $\varepsilon$-approximate coarse correlated equilibrium is better than what they got from their advice.

\begin{algorithm}

\DontPrintSemicolon
\KwIn{Advice strategy $\alpha^*_{\text{advice}}$, $\phi$-no-regret learner $A_{NR}$, phase length $\phi$, Tolerance $\varepsilon$}
\textbf{Initialization:} {$\texttt{phase} \leftarrow 0$, $R_{\text{NR}} \leftarrow 0$, $R_{\text{advice}} \leftarrow 0$}

\For{each round $t$}{
  \uIf{$\texttt{phase} = 0$ (Explore no-regret)}{
    Play action $\alpha_t \sim A_{NR}$ \\
    $R_{\text{NR}} \leftarrow R_{\text{NR}} + u_A(\alpha_t, \beta_t)$ \\
    After $\phi$ rounds, switch to phase 1
  }

  \uElseIf{$\texttt{phase} = 1$ (Explore advice)}{
    Play action $\alpha^*_{\text{advice}}$ \\
    $R_{\text{advice}} \leftarrow R_{\text{advice}} + u_A(\alpha^*_{\text{advice}}, \beta_t)$ \\
    After $\phi$ rounds, switch to phase 2
  }

  \uElseIf{$\texttt{phase} = 2$ (Compare)}{
    \uIf{$R_{\text{NR}} + \phi\cdot\varepsilon < R_{\text{advice}}$}{
      Commit to $\alpha^*_{\text{advice}}$ (phase 4)
    }\Else{
      Reset $R_{\text{NR}} \leftarrow 0$ and switch to phase 3 (Confirm no-regret)
    }
  }

  \uElseIf{$\texttt{phase} = 3$ (Confirm no-regret)}{
    Play $\alpha_t \sim A_{NR}$\\
    $R_{\text{NR}} \leftarrow R_{\text{NR}} + u_A(\alpha_t, \beta_t)$ \\
    After $\phi$ rounds: \\
    \Indp
      \uIf{$R_{\text{NR}} + \phi\cdot\varepsilon  < R_{\text{advice}}$}{
        Commit to $\alpha^*_{\text{advice}}$ (phase 4)
      }\Else{
        Commit to $A_{NR}$  (phase 4)
      }
    \Indm
  }

  \ElseIf{$\texttt{phase} = 4$ (Commit)}{
    \uIf{confirmed no-regret}{
      Play $\alpha_t \sim A_{NR}$, update as usual
    }\Else{
      Play $\alpha^*_{\text{advice}}$
    }
  }
}
\caption{Player A Meta-Strategy with Cumulative Utility Comparison}
\label{Alg_nr_dominate}
\end{algorithm}
    
\end{proof}
\begin{corollary}
\label{cor:no-regret-ext}
    There exists an advice-augmented algorithm for the optimizer for which when they engage with a $\phi$-no-adaptive-swap regret learner incurring at most $\varepsilon$ swap regret in each $\phi$ window, the optimizer's total expected utility weakly dominates utility from some $\varepsilon$-approximate correlated equilibrium.
\end{corollary}
\begin{proof}
    We can make use of the same proof as~\Cref{thm:no-regret} noting that we make use of no-swap-regret algorithms instead of no-regret algorithms.
 \end{proof}

\section{Details on Numerical Results}
\label{Exp_app}
In our evaluations, the park ranger takes the perspective of the optimizer and as such they seek to identify the Stackelberg strategy. The illegal poachers are constrained to have a no-swap-regret guarantee. We provide as advice the true Stackelberg strategy for the park rangers, a suboptimal strategy, and evaluate how the cumulative payoff and the probability with which the optimizer plays the advice or a no-swap-regret strategy evolves with time. 
We provide plots for the evolution of the cumulative payoff of the Park Ranger (optimizer) as they deploy the advice augmented strategy described in Theorem \ref{thm:no-regret}.

For this evaluation we use real-world data sourced from \cite{Movebank} which tracks the migratory patterns of 32 African elephants. We isolate one season of African elephant population density within this dataset (April - September 2011) and from there construct a 3 $\times$ 3 grid with the same number of elephants belonging to some grid. 

We define a repeated two-player game between the park ranger and poacher wherein each player selects a mixed strategy at each time step (a distribution over cells in the 3 $\times$ 3 grid world). 

Let $n$ be the number of elephants in a given grid cell, the payoff structure is defined as follows:

\begin{itemize}
    \item If the attacker and defender select the same cell, the defender receives $n$ (for successfully protecting elephants) while the defender receives $-2$ (a fixed legal penalty)
    \item If the attacker and defender select different cells, the attacker selects $n_{att}$(the number of elephants in the cell selected by the attacker) and the defender receives $-n_{att}$(the number of elephants attacked)
\end{itemize}

In our evaluations, we leverage payoffs defined by the following:
\begin{center}
\textbf{Distribution of Elephants}

\vspace{0.5em}
\renewcommand{\arraystretch}{1.5}
\begin{tabular}{|c|c|c|}
\hline
0 & 0 & 0 \\
\hline
0 & 1 & 1 \\
\hline
0 & 0 & 6 \\
\hline
\end{tabular}
\end{center}

\begin{table}[h!]
\centering
\begin{tabular}{cc}
\begin{minipage}{0.45\textwidth}
\centering
\textbf{Park Ranger Payoffs (Defender)}

\[
\begin{bmatrix}
  0 &  0 &  0 &  0 & -1 & -1 &  0 &  0 & -6 \\
  0 &  0 &  0 &  0 & -1 & -1 &  0 &  0 & -6 \\
  0 &  0 &  0 &  0 & -1 & -1 &  0 &  0 & -6 \\
  0 &  0 &  0 &  0 & -1 & -1 &  0 &  0 & -6 \\
  0 &  0 &  0 &  0 &  1 & -1 &  0 &  0 & -6 \\
  0 &  0 &  0 &  0 & -1 &  1 &  0 &  0 & -6 \\
  0 &  0 &  0 &  0 & -1 & -1 &  0 &  0 & -6 \\
  0 &  0 &  0 &  0 & -1 & -1 &  0 &  0 & -6 \\
  0 &  0 &  0 &  0 & -1 & -1 &  0 &  0 &  6 
\end{bmatrix}
\]
\end{minipage}
\end{tabular}
\end{table}

\begin{table}[h!]
\centering
\begin{tabular}{cc}
\begin{minipage}{0.45\textwidth}
\centering
\textbf{Illegal Poacher Payoffs (Attacker)}

\[
\begin{bmatrix}
 -2 &  0 &  0 &  0 &  1 &  1 &  0 &  0 &  6 \\
  0 & -2 &  0 &  0 &  1 &  1 &  0 &  0 &  6 \\
  0 &  0 & -2 &  0 &  1 &  1 &  0 &  0 &  6 \\
  0 &  0 &  0 & -2 &  1 &  1 &  0 &  0 &  6 \\
  0 &  0 &  0 &  0 & -2 &  1 &  0 &  0 &  6 \\
  0 &  0 &  0 &  0 &  1 & -2 &  0 &  0 &  6 \\
  0 &  0 &  0 &  0 &  1 &  1 & -2 &  0 &  6 \\
  0 &  0 &  0 &  0 &  1 &  1 &  0 & -2 &  6 \\
  0 &  0 &  0 &  0 &  1 &  1 &  0 &  0 & -2 
\end{bmatrix}
\]
\end{minipage}
\end{tabular}
\end{table}
Concretely here are the strategies we provide the optimizer:\\
\begin{align*}
&\textbf{True Stackelberg Strategy:}\\ 
&\begin{bmatrix}
0 & 0 & 0 & 0 & \frac{3}{19} & \frac{3}{19} & 0 & 0 & \frac{13}{19}
\end{bmatrix} \\
&\textbf{Incorrect Advice Stackelberg Strategy:} \\
&\begin{bmatrix}
0 & 0 & 0 & 0 & 0 & 0 & 0 & 0 & 1
\end{bmatrix}
\end{align*}

We also include for completeness additional evaluations for advice modes that extend beyond just an initial fixed strategy. In particular we also assume the following advice modes:
\begin{enumerate}
    \item Advice is sampled at random from the simplex. We sample $\mathbf{z} \sim \mathcal{N}(\mathbf{0}, I_n)$ and set $\alpha_i = e^{z_i}/\sum_{k=1}^n e^{z_k}$.
    \item Advice starts out as the true Stackelberg strategy then arbitrarily switches to the incorrect advice Stackelberg strategy.
    \item Advice starts out as the incorrect advice Stackelberg strategy then arbitrarily switches to the correct advice strategy.
\end{enumerate}

Note that in all these cases the advice is not fixed yet the construction of Algorithm~\ref{Alg_nr_dominate} takes in a fixed advice. To resolve this, we allow the player to build fixed advice from limited interaction with the advice mode. We allow this limited interaction to occur online and from there the agent essentially builds up an empirical estimate of fixed advice to use for the rest of the game. 

\begin{figure}[h]
  \centering
  \begin{subfigure}[t]{0.45\textwidth}
    \centering
    \includegraphics[width=\textwidth]{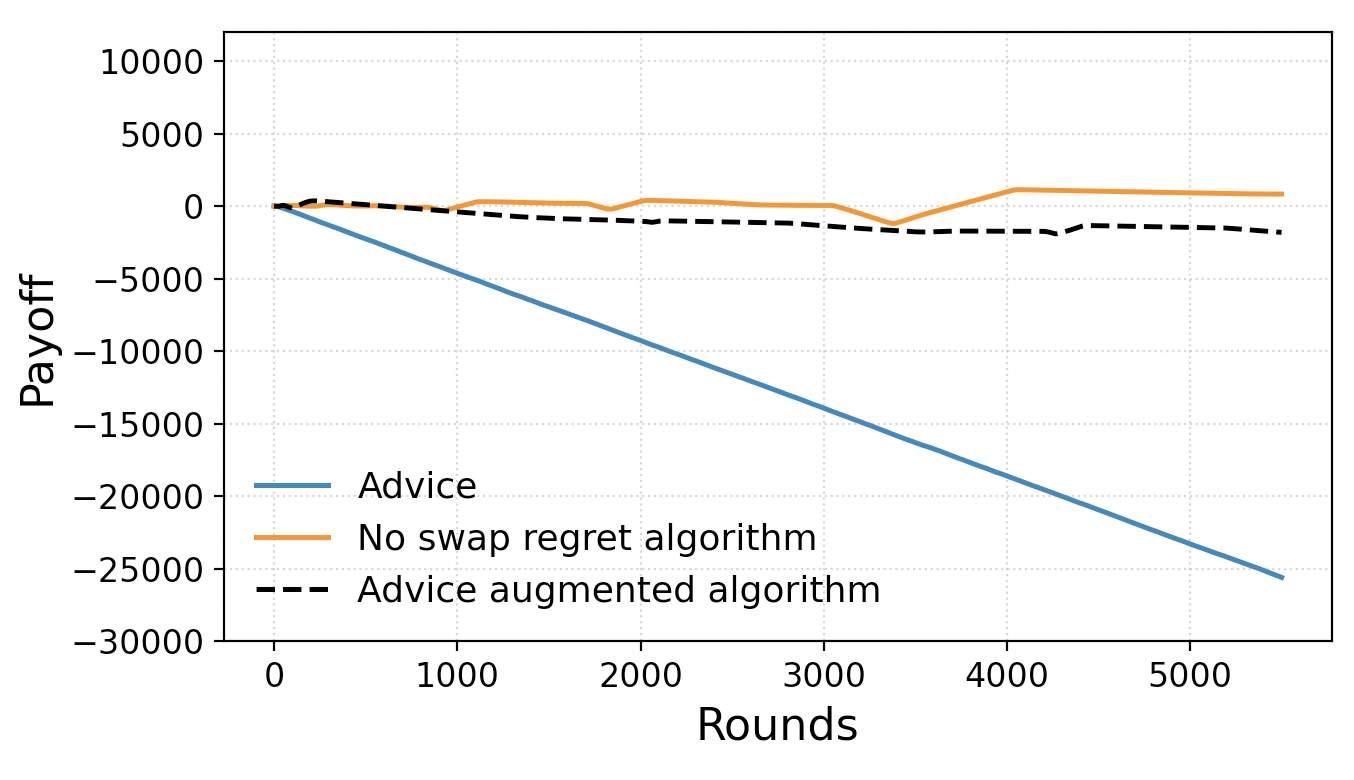}
    \caption{Optimizer cumulative payoff given Advice mode 1}
  \end{subfigure}
  \hfill
  \begin{subfigure}[t]{0.45\textwidth}
    \centering
    \includegraphics[width=\textwidth]{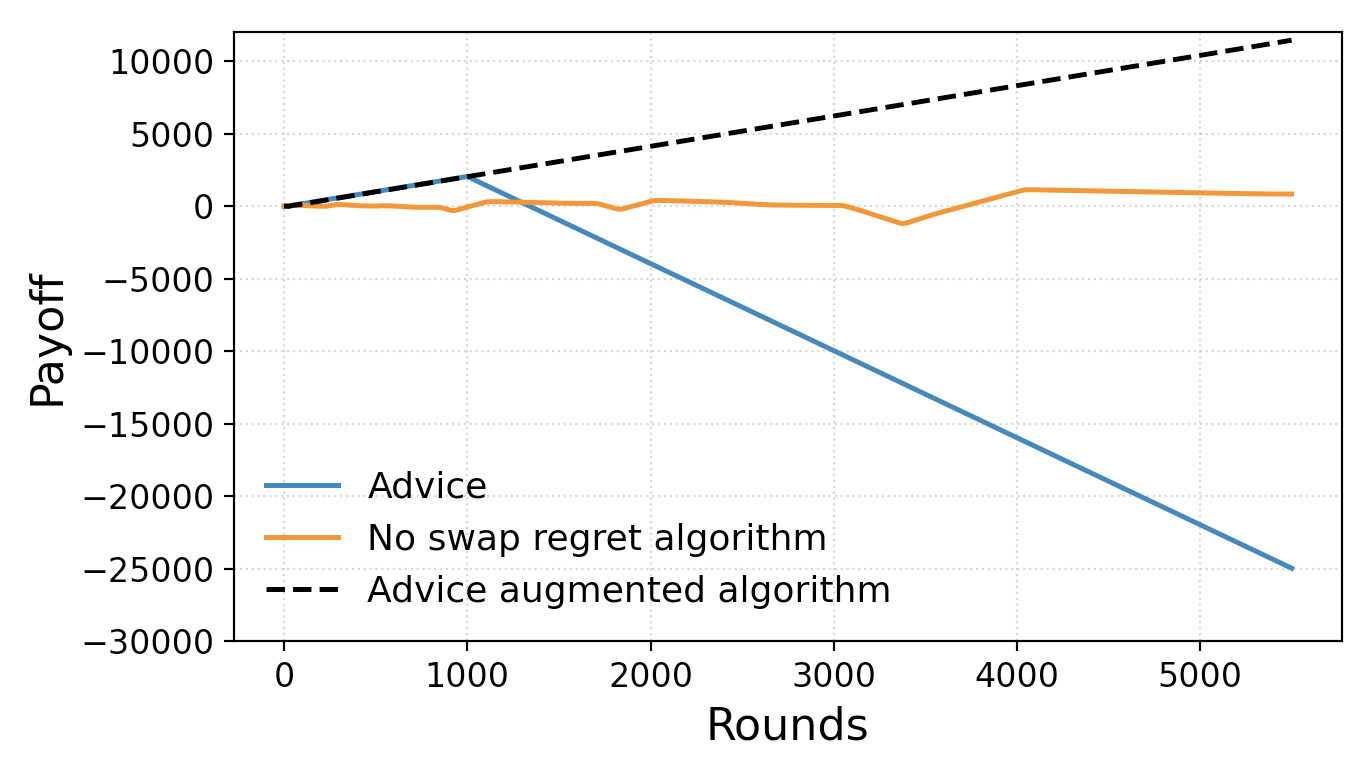}
    \caption{Optimizer cumulative payoff given Advice mode 2}
  \end{subfigure}
    \hfill
  \begin{subfigure}[t]{0.45\textwidth}
    \centering
    \includegraphics[width=\textwidth]{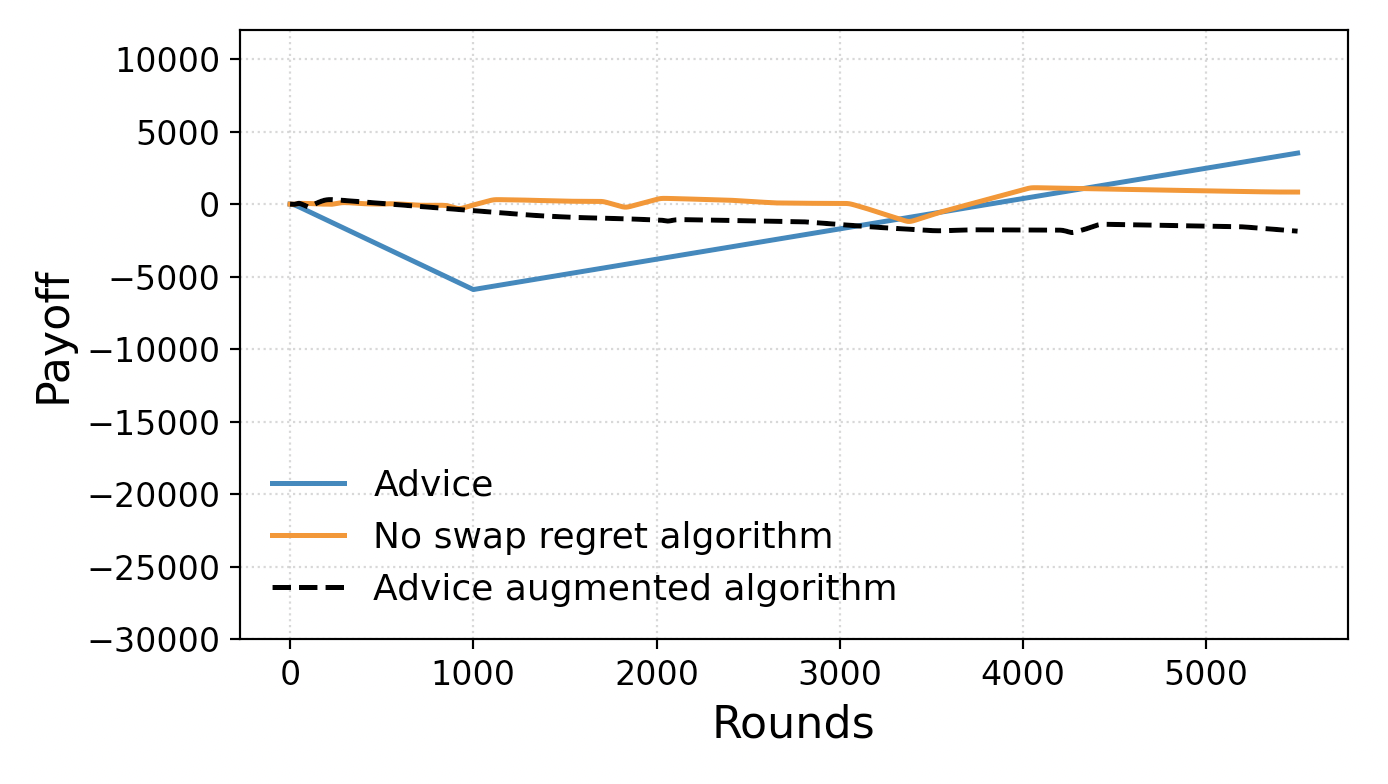}
    \caption{Optimizer cumulative payoff given Advice mode 3}
  \end{subfigure}
  \caption{Comparison of optimizer cumulative payoff under different advice instances.}
  \label{fig:additional_plots}
\end{figure}


\end{document}